\documentclass[11pt]{llncs}
\usepackage{amsmath,amssymb}
\usepackage{graphicx}
\newcommand{\comm}[1]{}

\usepackage{fullpage}
\usepackage[top=0.8in, bottom=0.9in, left=0.9in, right=0.9in]{geometry}

\usepackage{color,soul}
\usepackage{amsmath,amssymb}	
\usepackage[euler]{textgreek}
\usepackage{algorithm}
\usepackage{algpseudocode}
\usepackage{cite}
\usepackage{enumitem}
\graphicspath{{./Fig/}}

\usepackage[pdftex, plainpages = false, pdfpagelabels, 
                 bookmarks=false,
                 bookmarksopen = true,
                 bookmarksnumbered = true,
                 breaklinks = true,
                 linktocpage,
                 pagebackref,
                 colorlinks = true,  
                 linkcolor = blue,
                 urlcolor  = cyan,
                 citecolor = red,
                 anchorcolor = green,
                 hyperindex = true,
                 hyperfigures
                 ]{hyperref}

\def\calI{\mathcal{I}}

\def\calR{\mathcal{R}}
\def\calU{\mathcal{U}}

\def\calD{\mathcal{D}}
\def\hs{\hat{s}}
\def\hp{\hat{p}}

\def\fd{F\!D}

\def\opt{S_{\text{opt}}}

\newtheorem{observation}{Observation}

\begin{document}

\title{On the Line-Separable Unit-Disk Coverage and Related Problems\thanks{A preliminary version of this paper appeared in {\em Proceedings of the 34th International Symposium on Algorithms and Computation (ISAAC 2023)}. This research was supported in part by NSF under Grants CCF-2005323 and CCF-2300356.}}
\author{Gang Liu 
\and
Haitao Wang 
}

 \institute{
 Kahlert School of Computing\\
  University of Utah, Salt Lake City, UT 84112, USA\\
  \email{u0866264@utah.edu, haitao.wang@utah.edu}
}

\maketitle

\pagestyle{plain}
\pagenumbering{arabic}
\setcounter{page}{1}

\begin{abstract}
Given a set $P$ of $n$ points and a set $S$ of $m$ disks in the plane, the disk coverage
problem asks for a smallest subset of disks that together cover all points of $P$.
The problem is NP-hard. In this paper, we consider a line-separable unit-disk version of the problem where all disks have the same radius and their centers are separated from the points of $P$ by a line $\ell$. We present an $O((n+m)\log(n+m))$ time algorithm for the problem. This improves the previously best result of $O(nm+ n\log n)$ time. Our techniques also solve the line-constrained version of the problem, where centers of all disks of $S$ are located on a line $\ell$ while points of $P$ can be anywhere in the plane. Our algorithm runs in $O((n+m)\log (m+ n)+m \log m\log n)$ time, which improves the previously best result of $O(nm\log(m+n))$ time. 
In addition, our results lead to an algorithm of $O(n^3\log n)$ time for a half-plane coverage problem (given $n$ half-planes and $n$ points, find a smallest subset of half-planes covering all points); this improves the previously best algorithm of $O(n^4\log n)$ time. Further, if all half-planes are lower ones, our algorithm runs in $O(n\log n)$ time while the previously best algorithm takes $O(n^2\log n)$ time. 
\end{abstract}

\keywords{disk coverage, line-separable, unit-disk, line-constrained, half-planes}

\section{Introduction}
\label{sec:intro}

Given a set $P$ of $n$ points and a set $S$ of $m$ disks in the plane, the {\em disk coverage} problem asks for a smallest subset of disks such that every point of $P$ is covered by at least one disk in the subset. 
The problem is NP-hard, even if all disks have the same radius~\cite{ref:FederOp88,ref:MustafaIm10}.  
Polynomial time approximation algorithms have been proposed for the problem and many of its variants, e.g.,~\cite{ref:AgarwalNe20,ref:BusPr18,ref:ChanEx14,ref:ChanFa20,ref:GanjugunteGe11,ref:LiA15}. 

Polynomial time exact algorithms are known for certain special cases. If all points of $P$ are inside a strip bounded by two parallel lines and the centers of all disks lie outside the strip, then the problem is solvable in polynomial time~\cite{ref:AmbuhlCo06}.
If all disks of $S$ contain the same point, polynomial time algorithms also exist~\cite{ref:CalinescuSe04,ref:DasHo10}; in particular, applying the result in~\cite{ref:ChanEx14} (i.e., Corollary 1.7) yields an $O(mn^2(m+n))$ time algorithm.
In order to devise an efficient approximation algorithm for the general coverage problem (without any constraints), the {\em line-separable} version was considered in the literature~\cite{ref:AmbuhlCo06,ref:CarmiCo07,ref:ClaudeAn10}, where disk centers are separated from the points by a given line $\ell$. 
A polynomial time $4$-approximation algorithm is given in \cite{ref:CarmiCo07}. Amb\"uhl et al.~\cite{ref:AmbuhlCo06} derived an exact algorithm of $O(m^2n)$ time. 
An improved $O(nm+ n\log n )$ time algorithm is presented in \cite{ref:ClaudeAn10} and another algorithm in \cite{ref:PedersenAl22} runs in $O(n\log n+m^2\log n)$ in the worst case. 

The {\em line-constrained} version of the disk coverage problem has also been studied, where disk centers are on the $x$-axis while points of $P$ can be anywhere in the plane. Pedersen and Wang~\cite{ref:PedersenAl22} considered the weighted case in which each disk has a weight and the objective is to minimize the total weight of the disks in the subset that cover all points. Their algorithm runs in $O((m+n)\log(m+n) + \kappa\log m)$ time, where $\kappa$ is the number of pairs of disks that intersect and $\kappa=O(m^2)$ in the worst case. They reduced the runtime to $O((m+n)\log(m+n))$ for the {\em unit-disk case}, where all disks have the same radius, as well as the $L_{\infty}$ and $L_1$ cases, where the disks are squares and diamonds, respectively~\cite{ref:PedersenAl22}.
The 1D problem where disks become segments on a line and points are on the same line is also solvable in $O((m+n)\log(m+n))$~\cite{ref:PedersenAl22}.
Other types of line-constrained coverage problems have also been studied in the literature, e.g.,~\cite{ref:AltMi06,ref:BiloGe05,ref:BiniazFa18,ref:Lev-TovPo05}.

A related problem is when disks of $S$ are half-planes. 
For the weighted case, Chan and Grant \cite{ref:ChanEx14} proposed an algorithm for the lower-only case where all half-planes are lower ones; their algorithm runs in $O(n^4)$ time when $m=n$. 
With the observation that a half-plane may be considered as a unit disk of infinite radius, the techniques of~\cite{ref:PedersenAl22} solve the problem in $O(n^2\log n)$ time. 
For the general case where both upper and lower half-planes are present, Har-Peled and Lee \cite{ref:Har-PeledWe12} solved the problem in $O(n^5)$ time. Pedersen and Wang~\cite{ref:PedersenAl22} showed that the problem can be reduced to $O(n^2)$ instances of the lower-only case problem and thus can be solved in $O(n^4\log n)$ time.
To the best of our knowledge, we are not aware of any previous work particularly on the unweighted half-plane coverage problem. 

\subsection{Our result}

We assume that $\ell$ is the $x$-axis and all disk centers are below or on $\ell$ while all points of $P$ are above or on $\ell$. 
We consider the line-separable version of the disk coverage problem with the following {\em single-intersection condition}: For any two disks, their boundaries intersect at most once in the half-plane above $\ell$. Note that this condition is satisfied in both the unit-disk case (see Fig~\ref{fig:unitcase}) and the line-constrained case (see Fig.~\ref{fig:linecase}; more to explain below). Hence, an algorithm for this line-separable single-intersection case works for both the unit-disk case and the line-constrained case. Note that all problems considered in this paper are unweighted case in the $L_2$ metric. 

\begin{figure}[t]
\begin{minipage}[t]{0.49\textwidth}
\begin{center}
\includegraphics[height=1.0in]{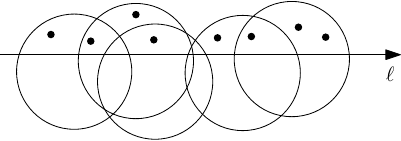}
\caption{\footnotesize Illustrating the line-separable unit-disk case.}
\label{fig:unitcase}
\end{center}
\end{minipage}
\hspace{0.05in}
\begin{minipage}[t]{0.49\textwidth}
\begin{center}
\includegraphics[height=1.0in]{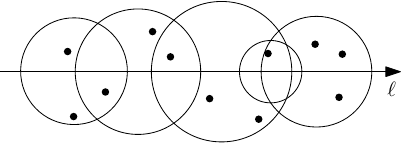}
\caption{\footnotesize Illustrating the line-constrained case (all disks are centred on $\ell$).}
\label{fig:linecase}
\end{center}
\end{minipage}
\vspace{-0.15in}
\end{figure}

For the above line-separable single-intersection problem, we give an algorithm of $O((n+m)\log (m+ n)+m \log m\log n)$ time in Section~\ref{sec:cover}. Based on observations, we find that some disks are ``useless'' and thus can be pruned from $S$. After pruning those useless disks, the remaining disks have certain properties so that we can reduce the problem to the 1D problem, which can then be easily solved. The overall algorithm is fairly simple conceptually. One challenge, however, is to show the correctness, that is, to prove why those ``useless'' disks are indeed useless. The proof is rather lengthy and technical. The bottleneck of the algorithm is to find those useless disks. 


\paragraph{\bf The line-constrained problem.}
Observe that the line-constrained problem where all disks of $S$ are centered on a line $\ell$ while points of $P$ can be anywhere in the plane is also a special case of the line-separable single-intersection problem. Indeed, for each point $p$ of $P$ below $\ell$, we could replace $p$ by its symmetric point with respect to $\ell$; in this way, we can obtain a set of points that are all above $\ell$. It is not difficult to see that an optimal solution using this new set of points is also an optimal solution for $P$. Furthermore, since disks are centered on $\ell$, although their radii may not be equal, the boundaries of any two disks intersect at most once above $\ell$. 
Hence, the problem is an instance of the line-separable single-intersection case. As such, applying our algorithm in Section~\ref{sec:cover} solves the line-constrained problem in $O((n+m)\log (m+ n)+m \log m\log n)$ time; this improves the previous algorithm in~\cite{ref:PedersenAl22}, which runs in $O(n\log n + m^2\log m)$ time in the worst case.

\paragraph{\bf The unit-disk case.}
To solve the line-separable unit-disk case, the algorithm in Section~\ref{sec:cover} still works. However, by making use of the property that all disks have the same radius, we further improve the runtime to $O((m+n)\log(m+n))$ in Section~\ref{sec:unitdisk}. 
This improves the $O(nm+ n\log n )$ time algorithm in \cite{ref:ClaudeAn10} and the $O(n\log n+m^2\log n)$ time algorithm in \cite{ref:PedersenAl22}.
The main idea of the improvement (over the algorithm in~Section \ref{sec:cover}) is to explore the property that all disks have the same radius. 

\paragraph{\bf The half-plane coverage problem.}
As in~\cite{ref:PedersenAl22}, our techniques also solve the half-plane coverage problem. Specifically, for the lower-only case, let $\ell$ be a horizontal line that is below all points of $P$. If we consider each half-plane as a unit disk of
infinite radius with center below $\ell$, then the problem becomes an instance of the line-separable unit-disk coverage problem. Therefore, applying our result leads to an algorithm of $O((m+n)\log(m+n))$ time. When $m=n$, this is $O(n\log n)$ time, improving the previous algorithm of $O(n^2\log n)$ time~\cite{ref:PedersenAl22}. For the general case where both the upper and lower half-plane are present, using the method in~\cite{ref:PedersenAl22} that reduces the problem to $O(n^2)$ instances of the lower-only case, the problem is now solvable in $O(n^2(m+n))\log(m+n))$ time.  When $m=n$, this is $O(n^3\log n)$ time, improving the previous algorithm of $O(n^4\log n)$ time~\cite{ref:PedersenAl22}.

\paragraph{\bf An algorithm in the algebraic decision tree model.} In the algebraic decision tree model, where only comparisons are counted towards the time complexity, combining with a technique recently developed by Chan and Zheng~\cite{ref:ChanHo23},
our method shows that the line-separable single-intersection problem (and thus the line-constrained problem) can be solved using $O((n+m)\log (n+m))$ comparisons. The details are presented at the end of Section~\ref{sec:cover}. In the following discussion, unless otherwise stated, all time complexities are measured in the standard real RAM model. 

\paragraph{\bf Remark.}
The results improve our original results in the conference version of this paper~\cite{ref:LiuOn23-conf}. The high-level algorithm framework (and its correctness proof) is the same as before, but this version provides more efficient algorithm implementations.

\section{Preliminaries}
\label{sec:pre}
In this section, we introduce some concepts and notations that we will use in the rest of the paper. 

We follow the notation defined in Section~\ref{sec:intro}, e.g., $P$, $S$, $m$, $n$, $\ell$. Without loss of generality, we assume that $\ell$ is the $x$-axis and points of $P$ are all above or on $\ell$ while centers of disks of $S$ are all below or on $\ell$. 
Under this setting, for each disk $s\in S$, only its portion above $\ell$ matters for our coverage problem. 
Hence, unless otherwise stated, a disk $s$ only refers to its portion above $\ell$. As such, the boundary of $s$ consists of an {\em upper arc}, i.e., the boundary arc of the original disk above $\ell$, and a {\em lower segment}, i.e., the intersection of $s$ with $\ell$. 
Notice that $s$ has a single leftmost (resp., rightmost) point, which is the left (resp., right) endpoint of the lower segment of $s$. 


We assume that each point of $P$ is covered by at least one disk since otherwise there would be no feasible solution. 
Our algorithm is able to check whether the assumption is met. 
We make a general position assumption that no point of $P$ lies on the boundary of a disk and no two points of $A$ have the same $x$-coordinate, where $A$ is the union of $P$ and the set of the leftmost and rightmost points of all disks. Degenerated cases can be easily handled by standard perturbation techniques, e.g., \cite{ref:EdelsbrunnerSi90}. 

For any point $p$ in the plane, we denote its $x$- and $y$-coordinates by $x(p)$ and $y(p)$, respectively.
We sort all points of $P$ in ascending order of their $x$-coordinates, resulting in a sorted list ${p_1, p_2,\cdots, p_n}$. 
We also sort all disks in ascending order of the $x$-coordinates of their leftmost points, resulting in a sorted list ${s_1, s_2,\cdots, s_m}$. We use $S[i,j]$ to denote the subset $\{s_i, s_{i+1},\cdots, s_j\}$; for convenience, $S[i,j]=\emptyset$ if $i>j$.
For each disk $s_i$, let $l_i$ and $r_i$ denote its leftmost and rightmost points, respectively.

For any disk $s$, we use $S_l(s)$ (resp., $S_r(s)$) to denote the set of disks $S$ whose leftmost points are left (resp., right) of that of $s$. 
As such, if the index of $s$ is $i$, then $S_l(s)=S[1,i-1]$ and $S_r(s)=S[i+1,m]$. If disk $s'\in S_l(s)$, then we also say that $s'$ is {\em to the left} of $s$; similarly, if $s'\in S_r(s)$, then $s'$ is {\em to the right} of $s$.

For a point $p_i\in P$ and a disk $s_k\in S$, we say that $p_i$ is {\em vertically above}  $s_k$ if $p_i$ is outside $s_k$ and $x(l_k) < x(p_i) < x(r_k)$. 

If $S'$ is a subset of $S$ that forms a coverage of $P$, then we call $S'$ a {\em feasible solution}. 
If $S'$ is a feasible solution of minimum size, then $S'$ is an {\em optimal solution}.

\paragraph{\bf The non-containment property.}
Suppose a disk $s_i$ contains another disk $s_j$. Then $s_j$ is redundant for our problem, since any point covered by $s_j$ is also covered by $s_i$. These redundant disks can be easily identified and removed from $S$ in $O(m\log m)$ time (indeed, this is a 1D problem by observing that $s_i$ contains $s_j$ if and only if the lower segment of $s_i$ contains that of $s_j$). Hence, to solve our problem, we first remove such redundant disks and work on the remaining disks. 
For simplicity, from now on we assume that no disk of $S$ contains another. Therefore, $S$ has the following {\em non-containment} property, which our algorithm relies on. 

\begin{observation}{\em (Non-Containment Property)}\label{obser:FIFO}
For any two disks $s_i, s_j \in S$, $x(l_i) < x(l_j)$ if and only if $x(r_i) < x(r_j)$.
\end{observation}

\section{The line-separable single-intersection case}
\label{sec:cover}

In this section, we present our algorithm for the disk coverage problem in the line-separable single-intersection case. We follow the notation defined in Section~\ref{sec:pre}. 

For each disk $s_i\in S$, we define two indices $a(i)$ and $b(i)$ of points of $P$ (where $p_{a(i)}$ and $p_{b(i)}$ are not contained in $s_i$), which are critical to our algorithm. 

\begin{definition}
\begin{itemize}
    \item 
    Among all points of $P$ covered by the union of the disks of $S[1,i-1]$ but not covered by $s_i$, define $a(i)$ to be the largest index of these points; if no such point exists, then let $a(i)=0$. 
    \item
    Among all points of $P$ covered by the union of the disks of $S[i+1,m]$ but not covered by $s_i$, define $b(i)$ to be the smallest index of these points; if no such point exists, then let $b(i)=n+1$.
\end{itemize}
\end{definition}


We say that a disk $s_i\in S$ is {\em prunable} if $a(i)\geq b(i)$. 

We now describe our algorithm. Although the algorithm description looks simple, it is quite challenging to prove the correctness; we devote Section~\ref{sec:covercorrect} to it. The implementation of the algorithm, which is also not trivial, is presented in Section~\ref{sec:coverimplement}. 

\paragraph{\bf Algorithm description.}
The algorithm has three main steps. 

\begin{enumerate}
    \item We first find all prunable disks of $S$. 
        Let $S^*$ denote the subset of disks of $S$ that are not prunable. We will prove in Section~\ref{sec:covercorrect} that $S^*$ contains an optimal solution for the coverage problem on $P$ and $S$. This means that it suffices to work on $S^*$ and $P$. 

    \item We then compute $a(i)$ and $b(i)$ for all disks $s_i\in S^*$. We will show in Section~\ref{sec:coverimplement} that this step together with the above first step for computing $S^*$ can be done in $O((n+m)\log (n+m)+m\log m\log n)$ time.  

    \item We reduce the disk coverage problem on $S^*$ and $P$ to a 1D coverage problem as follows. For each point of $P$, we project it vertically onto $\ell$. Let $P'$ be the set of all projected points. For each disk $s_i\in S^*$, we create a line segment on $\ell$ whose left endpoint has $x$-coordinate equal to $x(p_{a(i)+1})$ and whose right endpoint has $x$-coordinate equal to $x(p_{b(i)-1})$ (if $a(i)+1=b(i)$, then let the $x$-coordinate of the right endpoint be $x(p_{a(i)+1})$). Let $S'$ be the set of all the segments thus created. 

    We solve the following 1D coverage problem: Find a minimum subset of the segments of $S'$ that together cover all points of $P'$. This problem can be easily solved in $O((|S'|+|P'|)\log (|S'|+|P'|))$ time~\cite{ref:PedersenAl22},\footnote{The algorithm in \cite{ref:PedersenAl22}, which uses dynamic programming, is for the weighted case where each segment has a weight. Our problem is simpler since it is an unweighted case. We can use a simple greedy algorithm to solve it.} which is $O((m+n)\log(m+n))$ since $|P'|=n$ and $|S'|\leq m$. 

    Suppose $S'_1$ is any optimal solution to the above 1D coverage problem. We create a subset $S_1$ of $S^*$ as follows. For each segment of $S'_1$, suppose it is created from a disk $s_i\in S^*$; then we add $s_i$ to $S_1$. We will prove in Section~\ref{sec:covercorrect} that $S_1$ is an optimal solution to the coverage problem for $S^*$ and $P$.
\end{enumerate}

We summarize the result in the following theorem. 

\begin{theorem}\label{theo:coverage}
Given a set $P$ of $n$ points and a set $S$ of $m$ disks in the plane such that the disk centers are separated from points of $P$ by a line, and the single-intersection condition is satisfied, 
the disk coverage problem for $P$ and $S$ is solvable in $O((n+m)\log (m+ n)+m \log m\log n)$ time.
\end{theorem}

\paragraph{\bf The unit-disk case.} In Section~\ref{sec:unitdisk}, we will reduce the time to $O((n+m)\log (n+m))$ for the unit-disk case. The algorithm follows the above except that we implement the first two steps in a more efficient way (i.e., in $O((n+m)\log (n+m))$ time) by utilizing the property that all disks have the same radius. 

\subsection{Algorithm correctness}
\label{sec:covercorrect}

We now prove the correctness of our algorithm. Lemma~\ref{lem:coverprune} justifies the correctness of the first main step. To prove Lemma~\ref{lem:coverprune}, whose proof is lengthy and technical, we first prove the following Lemma~\ref{lem:10} (which will also be useful in our algorithm implementation in Section~\ref{sec:coverimplement} for finding all prunable disks). Recall the definition of $S_l(s)$ and $S_r(s)$ in Section~\ref{sec:pre}. 

\begin{lemma}\label{lem:10}
A disk $s$ is prunable if and only if there exists a point in $P$ that is outside $s$ but is covered by both a disk in $S_l(s)$ and a disk in $S_r(s)$. 
\end{lemma}
\begin{proof}
Let $i$ be the index of $s$, i.e., $s=s_i$. Our goal is to show that $s_i$ is prunable if and only if there exists a point in $P$ that is outside $s_i$ but is covered by both a disk in $S[1,i-1]$ and a disk in $S[i+1,m]$. 

\paragraph{\bf The ``if'' direction.}
Suppose $P$ has a point $p_t$ that is outside $s_i$ but is covered by a disk $s_k$ and a disk $s_j$ with $k<i<j$. Our goal is to prove that $a(i)\geq b(i)$, meaning that $s_i$ is prunable by definition. 
Since $s_k$ covers $p_t$ and $k<i$, by definition we have $a(i)\geq t$. On the other hand, since $s_j$ covers $p_t$ and $i<j$, by definition we have $b(i)\leq t$. As such, we obtain $a(i)\geq b(i)$. 

\paragraph{\bf The ``only if'' direction.}
Suppose $s_i$ is a prunable disk. Our goal is to show that there exists a point $p^*\in P$ that is outside $s_i$ but covered by both a disk in $S[1,i-1]$ and a disk in $S[i+1,m]$. 
Since $s_i$ is a prunable disk, we have $a(i)\geq b(i)$, and further, there are a disk $s_k$ with $k<i$ that covers $p_{a(i)}$ and a disk $s_j$ with $j>i$ that covers $p_{b(i)}$. Depending on whether $s_k$ covers $p_{b(i)}$, there are two cases.  

\begin{figure}[t]
\begin{minipage}[t]{0.49\textwidth}
\begin{center}
\includegraphics[height=1.0in]{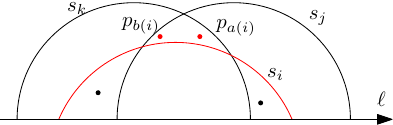}
\caption{\footnotesize Illustrating the case where $p_{b(i)}\in s_k$.}
\label{fig:prunable}
\end{center}
\end{minipage}
\hspace{0.05in}
\begin{minipage}[t]{0.49\textwidth}
\begin{center}
\includegraphics[height=1.0in]{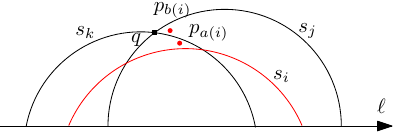}
\caption{\footnotesize Illustrating the case where $p_{b(i)}\not\in s_k$.}
\label{fig:prunable10}
\end{center}
\end{minipage}
\vspace{-0.15in}
\end{figure}

\begin{enumerate}
    \item 
    If $s_k$ covers $p_{b(i)}$ (see Fig.~\ref{fig:prunable}), then $p_{b(i)}$ is covered by both $s_k$ and $s_j$. Since $k<i$ and $j>i$, we can use $p_{b(i)}$ as our target point $p^*$. 
    \item 
    If $s_k$ does not cover $p_{b(i)}$ (see Fig.~\ref{fig:prunable10}), then since $s_k$ covers $p_{a(i)}$, we have $a(i)\neq b(i)$ and thus $a(i)>b(i)$. Since $k<j$, $s_k$ covers $p_{a(i)}$, and $s_j$ covers $p_{b(i)}$, due to the non-containment property of $S$, we have $x(l_k)<x(l_j)<x(p_{b(i)})< x(p_{a(i)})<x(r_k)<x(r_j)$, implying that the upper arcs of $s_k$ and $s_j$ must intersect, say, at at point $q$ (see Fig.~\ref{fig:prunable10}). As $p_{b(i)}\not\in s_k$, $p_{b(i)}$ must be vertically above $s_k$. This implies that $x(q)< x(p_{b(i)})$ must hold. Hence, the region of $s_k$ to the right of $q$ must be inside $s_j$. Since $x(q)<x(p_{b(i)})< x(p_{a(i)})$ and $p_{a(i)}$ is in $s_k$,  $p_{a(i)}$ must be in $s_j$ as well. Therefore, $p_{a(i)}$ is in both $s_k$ and $s_j$. As such, we can use $p_{a(i)}$ as our target point $p^*$.
\end{enumerate}

The lemma thus follows. 
\qed
\end{proof}

The following observation, which follows immediately from the non-containment property of $S$, is needed in the proof of Lemma~\ref{lem:coverprune}. 

\begin{observation}\label{obser:contain}
For any disk $s$ and a point $p$ outside $s$, if $p$ is covered by both a disk $s_i\in S_l(s)$ and a disk $s_j$ in $S_r(s)$, then $s\subseteq s_i\cup s_j$ (see Fig.~\ref{fig:coverunion}).     
\end{observation}

\begin{figure}[h]
\begin{minipage}[t]{\textwidth}
\begin{center}
\includegraphics[height=1.0in]{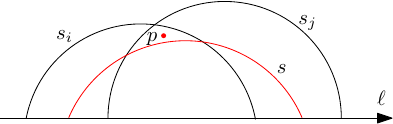}
\caption{\footnotesize Illustrating Observation~\ref{obser:contain}.}
\label{fig:coverunion}
\end{center}
\end{minipage}
\vspace{-0.15in}
\end{figure}


\begin{lemma}\label{lem:coverprune}
$S^*$ contains an optimal solution to the coverage problem on $S$ and $P$. 
\end{lemma}
\begin{proof}
Let $\opt$ be an optimal solution. Let $Q$ be the set of all prunable disks, i.e., $Q=S\setminus S^*$. If $\opt\cap Q=\emptyset$, then $\opt\subseteq S^*$ and thus the lemma trivially follows. In what follows, we assume that $|\opt\cap Q|\geq 1$. 

Pick an arbitrary point from $\opt\cap Q$, denoted $\hs_1$. Below, we give a process that can find a disk $s^*$ from $S^*$ to replace $\hs_1$ in $\opt$ such that the new set $\opt^1=\{s^*\}\cup \opt\setminus\{\hs_1\}$ still forms a coverage of $P$ (i.e., $\opt^1$ is a feasible solution), implying that $\opt^1$ is also an optimal solution since $|\opt^1|=|\opt|$. As $p^*\in S^*$, we have $|\opt^1\cap Q|=|\opt\cap Q|-1$. If $\opt^1\cap Q$ is still nonempty, then we can repeat the above process for other points in $\opt^1\cap Q$ until we obtain an optimal solution $\opt^*$ with $\opt^*\cap Q=\emptyset$, which will prove the lemma. 

We now give a process to find a target disk $s^*$. The process involves induction. To help the reader understand it better, we first provide details of the first two iterations of the process (we will introduce some notation that appears unnecessary for the first two iterations, but these will be needed when we describe the induction). 

\paragraph{\bf The first iteration.}
Let $\opt'=\opt\setminus\{\hs_1\}$. Since $\hs_1\in Q$, by Lemma~\ref{lem:10}, $P$ has a point $\hp_1$ outside $\hs_1$ but is covered by a disk $\hs_1^l\in S_l(\hs_1)$ and a disk $\hs_1^r\in S_r(\hs_1)$. By Observation~\ref{obser:contain}, $\hs_1\subseteq \hs_1^l\cup \hs_1^r$. Since $\hp_1$ is outside $\hs_1$ and $\opt = \opt'\cup \{\hs_1\}$ forms a coverage of $P$, $\opt'$ must have a disk $s$ that covers $\hp_1$. Clearly, $s$ is either in $S_l(\hs_1)$ or in $S_r(\hs_1)$. Without loss of generality, we assume that $s\in S_r(\hs_1)$. Since $\hs_1^r$ refers to an arbitrary disk of $S_r(\hs_1)$ that covers $\hp_1$ and $s$ is also a disk of $S_r(\hs_1)$ that covers $\hp_1$, for notational convenience, we let $\hs_1^r$ refer to $s$. As such, $\hs_1^r$ is in $\opt'$. 

Consider the disk $\hs_1^l$. Since $\hs_1\subseteq \hs_1^l\cup \hs_1^r$ and $\hs_1^r$ is in $\opt'$, it is not difficult to see that the area covered by the union of the disks of $\opt$ is contained in the area covered by the union of the disks of $\opt'\cup \{\hs_1^l\}$ and thus $\opt'\cup \{\hs_1^l\}$ is a feasible solution. As such, if $\hs_1^l\not\in Q$, then we can use $\hs_1^l$ as our target disk $s^*$ and our process (for finding $s^*$) is done. In what follows, we assume $\hs_1^l\in Q$. 

For any subset $S'$ of $S$, we define $\calR(S')$ as the area covered by the union of the disks of $S'$, i.e., $\calR(S')=\bigcup_{s\in S'}s$. 

We let $\hs_2=\hs_1^l$. Define $A_1=\{\hs_1^r\}$. According to the above discussion, we have $A_1\subseteq \opt'$, $\hs_1\subseteq \calR(A_1)\cup \hs_2$, and $\opt'\cup \{\hs_2\}$ is a feasible solution.  

\paragraph{\bf The second iteration.}
We are now entering the second iteration of our process. First notice that $\hs_2$ cannot be $\hs_1$ since $\hs_2=\hs_1^l$, which cannot be $\hs_1$. Our goal in this iteration is to find a {\em candidate disk} $s'$ to replace $\hs_2$ so that $\opt'\cup \{s'\}$ also forms a coverage of $P$. Consequently, if $s'\not\in Q$, then we can use $s'$ as our target $s^*$; otherwise, we need to guarantee $s'\neq \hs_1$ so that our process does not enter a dead loop. 
The discussion here is more involved than in the first iteration. 

Since $\hs_2\in Q$, by Lemma~\ref{lem:10}, $P$ has a point $\hp_2$ outside $\hs_2$ but is covered by a disk $\hs_2^l\in S_l(\hs_2)$ and a disk $\hs_2^r\in S_r(\hs_2)$. By Observation~\ref{obser:contain}, $\hs_2\subseteq \hs_2^l\cup \hs_2^r$.
Depending on whether $\hp_2$ is in $\calR(A_1)$, there are two cases. 

\begin{itemize}
    \item 
If $\hp_2\not\in \calR(A_1)$, then since $\hp_2\not\in \hs_2$ and $\hs_1\subseteq \calR(A_1)\cup \hs_2$, we obtain $\hp_2\not\in \hs_1$. We can now basically repeat our argument in the first iteration. 
Since $\hp_2$ is outside $\hs_2$ and $\opt'\cup \{\hs_2\}$ is a feasible solution, $\opt'$ must have a disk $s$ that covers $\hp_2$. Clearly, $s$ is either in $S_l(\hs_2)$ or in $S_r(\hs_2)$. Without loss of generality, we assume that $s\in S_r(\hs_2)$. Since $\hs_2^r$ refers to an arbitrary disk of $S_r(\hs_2)$ that covers $\hp_2$ and $s$ is also a disk of $S_r(\hs_2)$ that covers $\hp_2$, for notational convenience, we let $\hs_2^r$ refer to $s$. As such, $\hs_2^r$ is in $\opt'$. 

We let $\hs_2^l$ be our candidate disk, which satisfies our need as discussed above for $s'$. Indeed, since $\opt'\cup \{\hs_2\}$ is an optimal solution, $\hs_2\subseteq \hs_2^l\cup \hs_2^r$, and $\hs_2^r\in \opt'$, we obtain that $\opt'\cup \{\hs_2^l\}$ also forms a coverage of $P$. Furthermore, since $\hs_2^l$ contains $\hp_2$ while $\hs_1$ does not, we know that $\hs_2^l\neq \hs_1$. Therefore, if $\hs_2^l\not\in Q$, then we can use $\hs_2^l$ as our target $s^*$ and we are done with the process. Otherwise, we let $\hs_3=\hs_2^l$ and then enter the third iteration argument. In this case, we let $A_2=A_1\cup \{\hs_2^r\}$. According to our above discussion, we have $A_2\subseteq \opt'$, $\hs_2\subseteq \calR(A_2)\cup \hs_3$, and $\{\hs_3\}\cup \opt'$ is a feasible solution.  

\item 
If $\hp_2\in \calR(A_1)$, then we let $\hs_2^l$ be our candidate disk. We show in the following that it satisfies our need as discussed above for $s'$, i.e., $\{\hs_2^l\}\cup \opt'$ forms a coverage of $P$ and $\hs_2^l\neq \hs_1$. 

Indeed, since $A_1=\{\hs_1^r\}$ and $\hp_2\in \calR(A_1)$, $\hp_2$ is inside $\hs_1^r$. Since $\hs_1^r$ is to the right of $\hs_1$, and $\hs_2$, which is $\hs_1^l$, is to the left of $\hs_1$, we obtain that $\hs_1^r$ is to the right of $\hs_2$, i.e., $\hs_1^r\in S_r(\hs_2)$. 
Since $\hs_2^l$ contains $\hp_2$, $\hs_2^l\in S_l(\hs_2)$, $\hs_1^r$ contains $\hp_2$, and $\hs_1^r\in S_r(\hs_2)$, by Observation~\ref{obser:contain}, 
we obtain that $\hs_2\subseteq \hs_2^l\cup \hs_1^r$, i.e., $\hs_2\subseteq \hs_2^l\cup \calR(A_1)$. Since $\opt'\cup \{\hs_2\}$ is a feasible solution and $A_1\subseteq \opt'$, it follows that $\{\hs_2^l\}\cup \opt'$ is also a feasible solution. On the other hand, since $\hs_2^l$ is in $S_l(\hs_2)$ while $\hs_2$ (which is $\hs_1^l$) is in $S_l(\hs_1)$, we know that $\hs_2^l$ is in $S_l(\hs_1)$ and thus 
$\hs_2^l\neq \hs_1$. 

As such, if $\hs_2^l\not\in Q$, we can use $\hs_2^l$ as our target $s^*$ and we are done with the process. Otherwise, we let $\hs_3=\hs_2^l$ and continue on the third iteration. In this case, we let $A_2=A_1$. According to our above discussion, we have $A_2\subseteq \opt'$, $\hs_2\subseteq \calR(A_2)\cup \{\hs_3\}$, and $\{\hs_3\}\cup \opt'$ is a feasible solution.  
\end{itemize}

This finishes the second iteration of the process. 

\paragraph{\bf Inductive step.} In general, suppose that we are entering the $i$-th iteration of the process with disk $\hs_i\in Q$, $i\geq 2$. We make the following inductive hypothesis for $i$.

\begin{enumerate}
    \item We have disks $\hs_k\in Q$ for all $k=1,2,\ldots,i-1$ in the previous $i-1$ iterations such that $\hs_i\neq \hs_k$ for any $1\leq k\leq i-1$
    \item We have subsets $A_k$ for all $k=1,2,\ldots,i-1$ such that $A_1\subseteq A_2\subseteq\cdots \subseteq A_{i-1}\subseteq\opt'$, and $\hs_k\subseteq \calR(A_k)\cup \hs_{k+1}$ holds for each $1\leq k\leq i-1$. 
    \item For any $1\leq k\leq i$, $\{\hs_k\}\cup \opt'$ is a feasible solution.  
\end{enumerate}

Our above discussion showed that the hypothesis holds for $i=2$ and $i=3$. 
Next we discuss the $i$-th iteration for any general $i$. Our goal is to find a candidate disk $\hs_{i+1}$ so that $\opt'\cup\{\hs_{i+1}\}$ is a feasible solution and the inductive hypothesis holds for $i+1$. 

Since $\hs_i\in Q$, by Lemma~\ref{lem:10}, $P$ has a point $\hp_i$ outside $\hs_i$ but is covered by a disk $\hs_i^l\in S_l(\hs_i)$ and a disk $\hs_i^r\in S_r(\hs_i)$. 
By Observation~\ref{obser:contain}, $\hs_i\subseteq \hs_i^l\cup \hs_i^r$. 
Depending on whether $\hp_i$ is in $\calR(A_{i-1})$, there are two cases. 

\begin{enumerate}
    \item If $\hp_i\not\in \calR(A_{i-1})$, then since $\hp_i$ is outside $\hs_i$ and $\opt'\cup \{\hs_i\}$ is a feasible solution, $\opt'$ must have a disk $s$ that covers $\hp_i$. Clearly, $s$ is either in $S_l(\hs_i)$ or in $S_r(\hs_i)$. Without loss of generality, we assume that $s\in S_r(\hs_i)$. Since $\hs_i^r$ refers to an arbitrary disk of $S_r(\hs_i)$ that covers $\hp_i$ and $s$ is also a disk of $S_r(\hs_i)$ that covers $\hp_i$, for notational convenience, we let $\hs_i^r$ refer to $s$. As such, $\hs_i^r$ is in $\opt'$. 

    We let $\hs_{i+1}$ be $\hs_{i}^l$ and define $A_{i}=A_{i-1}\cup \{\hs_{i}^r\}$. We argue in the following that the inductive hypothesis holds. 

    \begin{itemize}
        \item 
        Indeed, since $\{\hs_i\}\cup \opt'$ is a feasible solution, $\hs_i\subseteq \hs_i^l\cup \hs_i^r$, $\hs_i^r\in \opt'$, and $\hs_{i+1}=\hs_{i}^l$, we obtain that $\{\hs_{i+1}\}\cup \opt'$ is a feasible solution. This proves the third statement of the hypothesis. 
\item
    Since $A_{i}=A_{i-1}\cup \{\hs_{i}^r\}$, $A_{i-1}\subseteq \opt'$ by the inductive hypothesis and $\hs_i^r\in \opt'$, we obtain $A_i\subseteq \opt'$. Also, since $\hs_i\subseteq \hs_i^l\cup \hs_i^r$, $\hs_i^r\in A_{i}$, and $\hs_{i+1}=\hs_{i}^l$, we have $\hs_i\subseteq \calR(A_i)\cup \hs_{i+1}$.
    This proves the second statement of the hypothesis. 
\item 
    For any disk $\hs_k$ with $1\leq k\leq i-1$, to prove the first statement of the hypothesis, we need to show that $\hs_k\neq \hs_{i+1}$. To this end, since $\hp_i\in \hs_{i+1}$, it suffices to show that $\hp_i\not\in \hs_k$. Indeed, by the inductive hypothesis, $\hs_k\subseteq \calR(A_k)\cup \hs_{k+1}$ and $\hs_{k+1}\subseteq \calR(A_{k+1})\cup \hs_{k+2}$. Hence, $\hs_k\subseteq \calR(A_k)\cup \calR(A_{k+1})\cup \hs_{k+2}$. As $\calR(A_k)\subseteq \calR(A_{k+1})$, we obtain $\hs_k\subseteq \calR(A_{k+1})\cup \hs_{k+2}$. Following the same argument, we can derive $\hs_k\subseteq \calR(A_{i-1})\cup \hs_{i}$. Now that $\hp_i\not\in \calR(A_{i-1})$ and $\hp_i\not\in \hs_i$, we obtain $\hp_i\not\in \hs_k$. 
\end{itemize}

\item 
If $\hp_i\in \calR(A_{i-1})$, then $\hp_i$ is covered by a disk of $A_{i-1}$, say $s$. As $\hp_i\not\in\hs_i$, $s\neq \hs_i$ and thus $s$ is in $S_l(\hs_i)$ or $S_r(\hs_i)$. Without loss of generality, we assume that $s\in S_r(\hs_i)$. 

We let $\hs_{i+1}$ be $\hs_{i}^l$ and define $A_{i}=A_{i-1}$. We show in the following that the inductive hypothesis holds. 
\begin{itemize}
    \item 
    Since $\hp_i$ is in both $s$ and $\hs_i^l$, $\hs_i^l\in S_l(\hs_i)$, and $s\in S_r(\hs_i)$, by Observation~\ref{obser:contain}, $\hs_i\subseteq \hs_i^l\cup s$.     
    Furthermore, since $\{\hs_i\}\cup \opt'$ is a feasible solution, $s\in A_{i-1}\subseteq \opt'$, and $\hs_{i+1}=\hs_{i}^l$, we obtain that $\{\hs_{i+1}\}\cup \opt'$ is also a feasible solution. This proves the third statement of the hypothesis.  

    \item 
    Since $A_{i-1}\subseteq \opt'$ by inductive hypothesis and $A_i=A_{i-1}$, we have $A_i\subseteq \opt'$. As discussed above, $\hs_i\subseteq \hs_i^l\cup s$. Since $s\in A_{i-1}=A_i$ and $\hs_{i+1}=\hs_{i}^l$, we obtain $\hs_i\subseteq \calR(A_i)\cup \hs_{i+1}$. This proves the second statement of the hypothesis.  

    \item 
     For any disk $\hs_k$ with $1\leq k\leq i-1$, to prove the first statement of the hypothesis, we need to show that $\hs_k\neq \hs_{i+1}$. By hypothesis, we know that $\hs_k\neq \hs_i$, implying that $\hs_k\in S_l(\hs_i)$ or $\hs_k\in S_r(\hs_i)$. If $\hs_k\in S_r(\hs_i)$,  since $\hs_{i+1}=\hs_i^l\in S_l(\hs_i)$, it is obviously true that $\hs_k\neq \hs_{i+1}$. In the following, we assume that $\hs_k\in S_l(\hs_i)$ and we will prove that $\hs_k$ does not contain $\hp_i$, which implies that $\hs_k\neq \hs_{i+1}$ as $\hp_i\in \hs_{i+1}$. 

     First of all, since $\hp_i$ is covered by both $s\in S_r(\hs_i)$ and $\hs_{i+1}\in S_l(\hs_i)$, it must hold that $x(\hat{l}_i)<x(\hp_i)<x(\hat{r}_i)$, where $\hat{l}_i$ and $\hat{r}_i$ are the left and right endpoints of the lower segment of $\hs_i$ (i.e., the segment $\hs_i\cap \ell$), respectively (see Fig.~\ref{fig:nocover}). 
     Hence, since $s\in S_r(\hs_i)$ and $\hp_i\in s$, the upper arcs of $s$ and $\hs_i$ must cross each other, say, at a point $q$. As $\hp_i$ is in $s$ but not in $\hs_i$, we have $x(q)<x(\hp_i)$.

     \begin{figure}[t]
\begin{minipage}[t]{\textwidth}
\begin{center}
\includegraphics[height=1.2in]{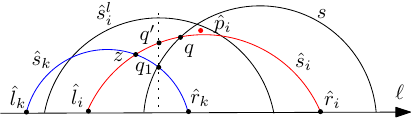}
\caption{\footnotesize Illustrating the proof of Lemma~\ref{lem:coverprune}. Note that $\hs_i^l$ is $\hs_{i+1}$.}
\label{fig:nocover}
\end{center}
\end{minipage}
\vspace{-0.15in}
\end{figure}

     Recall that $\{\hs_i\}\cup \opt'$ is a feasible solution. Also, $\opt'$ cannot be a feasible solution since that would contradict the fact that $\opt$ is an optimal solution as $|\opt|=|\opt'|+1$. This implies that $\hs_i$ is not contained in $\calR(\opt')$. Since all disk centers are below $\ell$, at least one point, say $q'$, on the upper arc of $\hs_i$ is not in $\calR(\opt')$. As $s\in \opt'$, $q'$ is not in $s$. Therefore, $x(q')<x(q)$ must hold. As $x(q)<x(\hp_i)$, we have $x(q')<x(\hp_i)$ (see Fig.~\ref{fig:nocover}). 

    Recall that our goal is to prove that $\hp_i\not\in \hs_k$. Let $\hat{l}_k$ and $\hat{r}_k$ be the left and right endpoints of the lower segment of $\hs_k$, respectively. If $x(\hat{r}_k)\leq x(q')$, then since $x(q')<x(\hp_i)$, it is obviously true that $\hs_k$ does not contain $\hp_i$. We therefore assume that $x(\hat{r}_k)> x(q')$  (see Fig.~\ref{fig:nocover}). Since $\hs_k\in S_l(\hs_i)$, we have $x(\hat{l}_k)<x(\hat{l}_i)$. As $x(\hat{l}_i)< x(q')$, we obtain $x(\hat{l}_k)<x(q')$. Since $x(\hat{l}_k)<x(q')<x(\hat{r}_k)$, the vertical line through $q'$ must intersect the upper arc of $\hs_k$ at a point, say, $q_1$ (see Fig.~\ref{fig:nocover}). 

    We claim $y(q_1)\leq y(q')$. Indeed, since $q'$ is not inside $\calR(\opt')$, $q'$ is on the upper envelope of all disks of $\opt'\cup \{\hs_i\}$, denoted by $\calU$. Recall that we have proved above (in the first case where $\hp_i\in \calR(A_{i-1})$) that $s_k\subseteq \calR(\opt')\cup \hs_i$. Therefore, the upper arc of $\hs_k$ cannot be higher than $\calU$. As $q'\in \calU$, it follows that $y(q_1)\leq y(q')$. 

    Since $x(\hat{l}_k)<x(\hat{l}_i)<x(q')<x(\hat{r}_k)$, due to the non-containment property, the upper arcs of $\hs_k$ and $\hs_i$ must cross each other at a single point, say $z$. Because $y(q_1)\leq y(q')$, we have $x(z)\leq x(q_1)$. As such, the region of $\hs_k$ to the right of $q_1$ must be inside $\hs_i$ (see Fig.~\ref{fig:nocover}). Recall that $x(q_1)=x(q')<x(\hp_i)$. As $\hp_i\not\in \hs_i$, $\hp_i$ cannot be in $\hs_k$. 
\end{itemize}
\end{enumerate}

This proves that the inductive hypothesis still holds for $i+1$. 

\bigskip

The inductive hypothesis implies that each iteration of the process always finds a new candidate disk $\hs_i$ such that $\opt'\cup \{\hs_i\}$ is a feasible solution. If $\hs_i\not\in Q$, then we can use $\hs_i$ as our target disk $s^*$ and we are done with the process. Otherwise, we continue with the next iteration. Since each iteration finds a new candidate disk (that was never used before) and $|Q|$ is finite, eventually we will find a candidate disk $\hs_i$ not in $Q$. 

This completes the proof of the lemma.  \qed
\end{proof}

It remains to prove the correctness of the third main step of our algorithm. 
For each disk $s_i\in S^*$, $a(i)< b(i)$ by definition; define $P(s_i)=\{p_j \ |\ a(i)<j<b(i)\}$.

\begin{lemma}\label{lem:30}
All points of $P(s_i)$ are inside $s_i$.    
\end{lemma}
\begin{proof}
 Assume to the contrary that a point $p_k\in P(s_i)$ is not in $s_i$. By definition, $a(i)< k< b(i)$. Recall that each point of $P$ is covered by a disk of $S$. Let $s$ be a disk of $S$ that covers $p_k$. Since $s\neq s_i$, $s$ is either in $S_l(s_i)$ or in $S_r(s_i)$. In the former case, by the definition of $a(i)$, $a(i)\geq k$, but this contradicts $a(i)<k$. In the latter case, by the definition of $b(i)$, $b(i)\leq k$; but this contradicts $k<b(i)$. \qed
\end{proof}

The following lemma justifies the correctness of the third main step of our algorithm.  

\begin{lemma}\label{lem:40}
Suppose $\opt$ is an optimal solution to the coverage problem on $S^*$ and $P$, and $s_i$ is a disk in $\opt$. Then, any point of $P\setminus P(s_i)$ must be covered by a disk of $\opt\setminus \{s_i\}$. 
\end{lemma}
\begin{proof}
Let $p$ be a point in $P\setminus P(s_i)$. If $p\not\in s_i$, then since the disks of $\opt$ form a coverage of $P$, there must be a disk of $\opt\setminus\{s_i\}$ that covers $p$. In the following, we assume that $p\in s_i$. Since $p\not\in P(s_i)$, by Lemma~\ref{lem:30}, either $x(p)\leq x(p_{a(i)})$ or $x(p)\geq x(p_{b(i)})$. In the following we only discuss the former case as the latter case is symmetric. 

Since $\opt$ is an optimal solution, $\opt$ must have a disk $s$ that covers $p_{a(i)}$ (see Fig.~\ref{fig:cover}). By definition, $p_{a(i)}$ is not in $s_i$. Hence, $s\neq s_i$ and thus $s$ is either in $S_l(s_i)$ or in $S_r(s_i)$. We claim that $s$ must be in $S_l(s_i)$. Indeed, assume to the contrary that $s\in S_r(s_i)$. Then, by the definition of $b(i)$, $b(i)\leq a(i)$ must hold, which contradicts $a(i)<b(i)$. 
Since $s\in S_l(s_i)$, we next prove that $s$ must cover $p$, which will prove the lemma. 

\begin{figure}[t]
\begin{minipage}[t]{\textwidth}
\begin{center}
\includegraphics[height=0.8in]{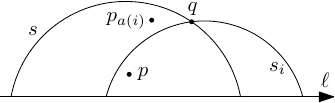}
\caption{\footnotesize Illustrating the proof of Lemma~\ref{lem:40}.}
\label{fig:cover}
\end{center}
\end{minipage}
\vspace{-0.15in}
\end{figure}

Indeed, since $x(p)\leq x(p_{a(i)})$, $p$ is inside $s_i$, and $s\in S_l(s_i)$,  due to the non-containment property of $S$, the upper arcs of $s_i$ and $s$ must intersect at a single point, say, $q$ (see Fig.~\ref{fig:cover}). Furthermore, since $p_{a(i)}$ is in $s$ but not in $s_i$, $x(p_{a(i)})\leq x(q)$ must hold. This implies that the region of $s_i$ left of $p_{a(i)}$ is inside $s$. As $p$ is in $s_i$ and $x(p)\leq x(p_{a(i)})$, $p$ must be inside $s$.
\qed
\end{proof}

In light of the preceding two lemmas, when considering the coverage of $s_i$, it suffices to consider only the points in $P(i)$. This establishes the correctness of the third main step of our algorithm. 

\subsection{Algorithm implementation}
\label{sec:coverimplement}
In this section, we describe the implementation of our algorithm. Recall that points of $P$ are indexed in ascending order of their $x$-coordinates as $p_1,p_2,\ldots,p_n$, and disks of $S$ are indexed in ascending order of their leftmost points as $s_1,s_2,\ldots,s_m$  (which is also the order of their rightmost points due to the non-containment property). 
Also recall that all points of $P$ are above the $x$-axis $\ell$ while the centers of all disks of $S$ are below $\ell$. 

To implement the algorithm, based on our algorithm description, it remains to perform the following two tasks: (1) Compute the subset $S^*\subseteq S$ of disks that are not prunable; (2) for each disk $s_i\in S^*$, compute $a(i)$ and $b(i)$. 

To achieve the above (1), we resort to Lemma~\ref{lem:10}. To this end, for each point $p\in P$, we define $\sigma_1(p) $ as the smallest index of the disk in $ S $ that covers $ p $, and $\sigma_2(p) $ as the largest index of the disk in $ S $ that covers $ p $. Note that both $\sigma_1(p)$ and $\sigma_2(p)$ are well defined since every point in $P$ is covered by at least one disk.
We first compute $\sigma_1(p)$ and $\sigma_2(p)$ in the following lemma. 

\begin{lemma}\label{lem:50}
Computing $\sigma_1(p)$ and $\sigma_2(p)$ for all points $p\in P$ can be done in $O((n+m)\log m)$ time. 
\end{lemma}
\begin{proof}
We only discuss how to compute $\sigma_1(p)$ since the algorithm for $\sigma_2(p)$ is similar. 

Let $T$ be a complete binary search tree whose leaves from left to right correspond to disks in their index order. Since $m=|S|$, the height of $T$ is $O(\log m)$. For each node $v\in T$, let $S_v$ denote the subset of disks of $S$ in the leaves of the subtree rooted at $v$. We use $\xi_v$ to denote the upper envelope of the $x$-axis $\ell$ and the upper arcs of all disks of $S_v$. Since all points of $P$ are above $\ell$, our algorithm is based on the observation that a point $p\in P$ is inside a disk of $S_v$ if and only if $p$ is below $\xi_v$. We construct $\xi_v$ for every node $v\in T$. Due to the single-intersection condition that the upper arcs of every two disks of $S$ intersect at most once, $\xi_v$ has at most $O(|S_v|)$ vertices. To see this, we can view the upper envelope of each upper arc of $S_v$ and $\ell$ as an extended arc. Every two such extended arcs still cross each other at most once and therefore their upper envelope has $O(|S_v|)$ vertices following the standard Davenport-Schinzel sequence argument~\cite{ref:SharirDa96} (see also \cite[Lemma 3]{ref:ChanAl16} for a similar problem). Notice that $\xi_v$ is exactly the upper envelope of these extended arcs and thus $\xi_v$ has $O(|S_v|)$ vertices. Note also that $\xi_v$ is $x$-monotone. In addition, given $\xi_u$ and $\xi_w$, where $u$ and $w$ are the two children of $v$, $\xi_v$ can be computed in $O(|S_v|)$ time by a straightforward line sweeping algorithm. As such, computing $\xi_v$ for all nodes $v\in T$ can be accomplished in time linear in $\sum_{v\in T}|S_v|$, which is $O(m\log m)$. 

Consider a point $p\in P$. We compute $\sigma_1(p)$ using $T$, as follows. Starting from the root of $T$, for each node $v$, we do the following. Let $u$ and $w$ be the left and right children of $v$, respectively. We first determine whether $p$ is below $\xi_u$; since $|S_v|\leq m$, this can be done in $O(\log m)$ time by binary search on the sorted list of the vertices of $\xi_u$ by their $x$-coordinates. If $p$ is below $\xi_u$, then $p$ must be inside a disk of $S_u$; in this case, we proceed with $v=u$. Otherwise, we proceed with $v=w$. In this way, $\sigma_1(p)$ can be computed after searching a root-to-leaf path of $T$, which has $O(\log m)$ nodes as the height of $T$ is $O(\log m)$. Because we spend $O(\log m)$ time on each node, the total time for computing $\sigma_1(p)$ is $O(\log^2 m)$. The time can be improved to $O(\log m)$ using fractional cascading~\cite{ref:ChazelleFr86}, as follows. 

The $x$-coordinates of all vertices of the upper envelope $\xi_v$ of each node $v\in T$ partition the $x$-axis into a set $\calI_v$ of intervals. To determine whether $p$ is below $\xi_v$, it suffices to find the interval of $\calI_v$ that contains $x(p)$, the $x$-coordinate of $p$ (after the interval is known, whether $p$ is below $\xi_v$ can be determined in $O(1)$ time). We construct a fractional cascading data structure on the intervals of $\calI_v$ of all nodes $v\in T$, which takes $O(m\log m)$ time~\cite{ref:ChazelleFr86} since the total number of such intervals is $O(m\log m)$. 

With the fractional cascading data structure, we only need to do binary search on the set of the intervals stored at the root to determine the interval containing $x(p)$, which takes $O(\log m)$ time. After that, for each node $u$ during the algorithm described above, the interval of $\calI_u$ containing $x(p)$ can be determined in $O(1)$ time~\cite{ref:ChazelleFr86}. As such, computing $\sigma_1(p)$ takes $O(\log m)$ time. 

In summary, computing $\sigma_1(p)$ for all points $p\in P$ takes $O((n+m)\log m)$ time in total. 
\qed
\end{proof}

We now describe our algorithm to compute $S^*$, or alternatively, find all prunable disks of $S$. By Lemma~\ref{lem:10}, we have the following observation. 

\begin{observation}\label{obser:30}
A disk $s_i\in S$ is prunable if and only there exists a point $p\in P$ such that $p\not\in s_i$ and $\sigma_1(p)\leq i\leq \sigma_2(p)$.    
\end{observation}
\begin{proof}
Suppose that $s_i$ is prunable. Then, by Lemma~\ref{lem:10}, there exists a point $p\in P$ such that $p\not\in s_i$ and $p$ is covered by both a disk in $S_l(s_i)$ and a disk in $S_r(s_i)$. By definition, we have $\sigma_1(p)<i<\sigma_2(p)$. 

On the other hand, suppose that there exists a point $p\in P$ such that $p\not\in s_i$ and $\sigma_1(p)\leq i\leq \sigma_2(p)$. Then, since $p\not\in s_i$ and $\sigma_1(p)\leq i$, we obtain $\sigma_1(p)\neq i$ and thus $\sigma_1(p)<i$. As such, $p$ is covered by a disk in $S_l(s_i)$. By a similar argument, $p$ is also covered by a disk in $S_r(s_i)$. By Lemma~\ref{lem:10}, $s_i$ is prunable. \qed
\end{proof}

With Observation~\ref{obser:30} at hand, the following lemma computes $S^*$. 

\begin{lemma}\label{lem:60}
All prunable disks of $S$ can be found in $O(n\log n+n\log m+m\log m\log n)$ time. 
\end{lemma}
\begin{proof}
Let $T$ be the standard segment tree~\cite[Section 10.3]{ref:deBergCo08} on the indices $1,2,\ldots,m$ by considering each index $i$ a point on the $x$-axis $\ell$ whose $x$-coordinate is $i$. The height of $T$ is $O(\log m)$. For each point $p\in P$, let $I_p$ denote the interval $[\sigma_1(p),\sigma_2(p)]$ of $\ell$. Following the definition of the standard segment tree~\cite[Section 10.3]{ref:deBergCo08}, we store $I_p$ in $O(\log m)$ nodes of $T$. For each node $v\in T$, let $P_v$ denote the subset of points $p$ of $P$ whose interval $I_p$ is stored at $v$. As such, $\sum_{v\in T}|P_v|=O(n\log m)$. 


Consider a disk $s_i\in S$. $T$ has a leaf corresponding to the index $i$, called {\em leaf-$i$}. Let $\pi_i$ denote the path of $T$ from the root to leaf-$i$. Following the definition of the segment tree, we have the following observation: $\bigcup_{v\in \pi_i}P_v=\{p\ |\ p\in P, \sigma_1(p)\leq i\leq \sigma_2(p)\}$. By Observation~\ref{obser:30}, to determine whether $s_i$ is prunable, it suffices to determine whether there is a node $v\in \pi_i$ such that $P_v$ has a point $p\not\in s_i$. Based on this observation, we perform the following processing on $T$. 
For each node $v\in T$, we construct a farthest Voronoi diagram for $P_v$, denoted by $\fd_v$, and  
then build a point location data structure on $\fd_v$ so that each point location query can be answered in $O(\log n)$ time~\cite{ref:EdelsbrunnerOp86,ref:KirkpatrickOp83}. Computing $\fd_v$ can be done in $O(|P_v|\log |P_v|)$ time~\cite{ref:ShamosCl75} and constructing the point location data structure takes $O(|P_v|)$ time~\cite{ref:EdelsbrunnerOp86,ref:KirkpatrickOp83}. Since $\sum_{v\in T}|P_v|=O(n\log m)$, the total time for constructing the farthest Voronoi diagrams for all nodes of $T$ is $O(n\log m\log n)$ and the total time for building all point location data structures is $O(n\log m)$. 

For each disk $s_i\in S$, for each node $v\in \pi_i$, 
using a point location query on $\fd_v$, we find the farthest point $p$ of $P_v$ from the center of $s_i$ in $O(\log n)$ time. If $p\not\in s_i$, then we assert that $s_i$ is prunable and halt the algorithm for $s_i$; otherwise, all points of $P_v$ are inside $s_i$ (and thus no point of $P_v$ can cause $s_i$ to be prunable) and we continue on other nodes of $\pi_i$. 
Since $\pi_i$ has $O(\log m)$ nodes and each point location query takes $O(\log n)$ time, it takes $O(\log m\log n)$ time to determine whether $s_i$ is prunable. Therefore, the total time for doing this for all disks $s_i\in S$ is $O(m\log m\log n)$. 

As such, we can find all prunable disks in a total of $O(m\log m\log n+ n\log m\log n+n\log m)$ time. Note that the factor $O(n\log n\log m)$ is due to the construction of all the farthest Voronoi diagrams $\fd_v$ for all nodes $v\in T$. We can further reduce the time to $O(n\log n+n\log m)$, as follows. 

First, $\fd_v$ is determined only by the points of $P_v$ that are vertices of the convex hull $H_v$ of $P_v$~\cite{ref:ShamosCl75}. Furthermore, once $H_v$ is available, $\fd_v$ can be computed in $O(|H_v|)$ time~\cite{ref:AggarwalA89}. On the other hand, $H_v$ can be obtained in $O(|P_v|)$ time if the points of $P_v$ are sorted by their $x$-coordinates. To have sorted lists for $P_v$ for all nodes $v\in T$, we do the following. At the start of the algorithm, we sort all points of $P$ by their $x$-coordinates in $O(n\log n)$ time. Then, for each point $p$ of $P$ following this sorted order, we find the nodes $v$ of $T$ where the interval $I_p$ should be stored and add $p$ to $P_v$, which takes $O(\log m)$ time~\cite[Section 7.4]{ref:deBergCo08}. In this way, after all points of $P$ are processed as above, $P_v$ for every node $v\in T$ is automatically sorted. As such, the total time for constructing all farthest Voronoi diagrams $\fd_v$ of all nodes $v\in T$ is $O(n\log n+n\log m)$. Therefore, the total time of the overall algorithm is $O(n\log n+n\log m+m\log m\log n)$. \qed
\end{proof}

The remaining task is to compute $a(i)$ and $b(i)$ for all disks $s_i\in S^*$. In what follows, we only discuss how to compute $a(i)$ since the algorithm for $b(i)$ is analogous. Our algorithm relies on the following observation, whose proof is based on the fact that none of the disks of $S^*$ is prunable. 

\begin{observation}\label{obser:40}
For each disk $s_i\in S^*$, $a(i)$ is the largest index among the points $p\in P$ with $\sigma_2(p)<i$. 
\end{observation}
\begin{proof}
First, consider a point $p_j\in P$ with $\sigma_2(p_j)<i$. Since $\sigma_2(p_j)<i$, we know that $p_j\not\in s_i$ and $p_j$ is covered by a point in $S[1,i-1]$. By the definition of $a(i)$, we have $j\leq a(i)$. 

On the other hand, let $j=a(i)$. By definition, $p_j\not\in s_i$ and $p_j$ is covered by a point in $S[1,i-1]$. By the definition of $\sigma_1(p_j)$, $\sigma_1(p_j)<i$ holds. We claim $\sigma_2(p_j)<i$. Indeed, assume to the contrary that $\sigma_2(p_j)\geq i$. Then, we have $\sigma_1(p_j)<i\leq \sigma_2(p_j)$. Since $p_j\not\in s_i$, by Observation~\ref{obser:30}, $s_i$ is prunable, which contradicts the fact that none of the disks of $S^*$ is prunable. As such, we obtain $\sigma_2(p_j)<i$. 

The above discussions combined lead to the observation. \qed
\end{proof}

In light of Observation~\ref{obser:40}, we have the following lemma. 

\begin{lemma}\label{lem:70}
$a(i)$ for all disks $s_i\in S^*$ can be computed in $O(n\log n+m)$ time. 
\end{lemma}
\begin{proof}
Recall that $p_1,p_2,\ldots,p_n$ are points of $P$ sorted in ascending order by $x$-coordinate. For each point $p\in P$, let $j(p)$ denote its index, referred to as the {\em $x$-sorted index}. We sort all the points $p$ of $P$ in ascending order by their values $\sigma_2(p)$ as $p^1,p^2,\ldots,p^n$. We process these points in this order. Each point $p^k$ is processed as follows. Our algorithm maintains $j^{k-1}$, the largest $x$-sorted index among all points $p^1,p^2,\ldots,p^{k-1}$. We first set $j^k$ to be the larger one of $j^{k-1}$ and $j(p^k)$. Then, for each integer $i$ with $\sigma_2(p^k)<i\leq \sigma_2(p^{k+1})$ (if $k=n$, then we consider each $i$ with $\sigma_2(p^k)<i\leq m$), if $s_i\in S^*$, then we set $a(i)=j^k$, whose correctness follows from Observation~\ref{obser:40}.
After all points of $P$ are processed as above, $a(i)$ for all disks $s_i\in S^*$ are computed. The total time is $O(n\log n+m)$. \qed
\end{proof}

Combining Lemmas~\ref{lem:50}, \ref{lem:60}, and \ref{lem:70}, the total runtime of the algorithm is $O((n+m)\log(n+m)+m\log m\log n)$. This proves Theorem~\ref{theo:coverage}. 

\paragraph{\bf An algebraic decision tree algorithm.} In the algebraic decision tree model, where the time complexity is measured only by the number of comparisons, we can solve the problem in $O((n+m)\log (n+m))$ time, i.e., using $O((n+m)\log (n+m))$ comparisons. For this, observe that the overall algorithm excluding Lemma~\ref{lem:60} runs in $O((n+m)\log(n+m))$ time. Hence, it suffices to show that Lemma~\ref{lem:60} can be solved using $O((n+m)\log (n+m))$ comparisons. To this end, notice that the factor $O(m\log m\log n)$ in the algorithm of Lemma~\ref{lem:60} is due to the point location queries on the farthest Voronoi diagrams $\fd_v$. There are a total of $O(m\log m)$ queries. The total combinatorial complexity of the diagrams $\fd_v$ of all nodes $v\in T$ is $O(n\log m)$. To solve these point location queries, we can use a technique recently developed by Chan and Zheng~\cite{ref:ChanHo23}. Specifically, we can simply apply \cite[Theorem~7.2]{ref:ChanHo23} to solve all our point location queries using $O((n+m)\log (n+m))$ comparisons (indeed, following the notation in \cite[Theorem~7.2]{ref:ChanHo23}, we have $t=O(m)$, $L=O(n\log m)$, $M=O(m\log m)$, and $N=O(n+m)$ in our problem; according to the theorem, all point location queries can be answered using $O(L+M+N\log N)$ comparisons, which is $O((n+m)\log (n+m))$).

\section{The unit-disk case}
\label{sec:unitdisk}

In this section, we show that the runtime of the algorithm can be reduced to $O((n+m)\log(n+m))$ for the unit-disk case. For this, observe that except for Lemma~\ref{lem:60}, the algorithm in Section~\ref{sec:cover} runs in $O((n+m)\log(n+m))$ time. Hence, 
it suffices to show that Lemma~\ref{lem:60} can be implemented in $O((n+m)\log(n+m))$ time for the unit-disk case. To this end, we have the following lemma. 

\begin{lemma}\label{lem:80}
For the unit-disk case, all prunable disks of $S$ can be found in $O((n+m)\log(n+m))$ time. 
\end{lemma}
\begin{proof}
We follow the notation in Lemma~\ref{lem:60}. The algorithmic scheme is similar to Lemma~\ref{lem:60}. The major change is that, instead of constructing the farthest Voronoi diagrams for the nodes of $T$, here we perform different processing by exploring the property that all disks of $S$ have the same radius. 

Consider a disk $s_i\in S$. To determine whether $s_i$ is prunable, recall that it suffices to decide for each node $v\in \pi_i$ whether $P_v$ has a point outside $s_i$. We explore the property that all disks of $S$ have the same radius (a disk of that radius is called a {\em unit disk}). For each point $p\in P$, let $D_p$ denote a unit disk centered at $p$. Define $\calD_v=\{D_p\ |\ p\in P_v\}$. Let $c_i$ denote the center of $s_i$. Observe that $P_v$ has a point outside $s_i$ if and only if $c_i$ is outside a disk of $\calD_v$. Recall that all points of $P$ are above the $x$-axis $\ell$ while the centers of all disks of $S$ are below $\ell$. Let $C_v$ denote the common intersection of all disks of $\calD_v$ below $\ell$. Observe that $c_i$ is outside a disk of $\calD_v$ if and only if $c_i$ is outside $C_v$. 

We perform the following processing on $T$. For each node $v\in T$, we construct $C_v$. Since all disks of $\calD_v$ have the same radius and their centers are all above $\ell$, the boundaries of every two disks intersect at most once in the halfplane below $\ell$. Using this property, $C_v$ can be computed in linear time by adapting Graham's scan after $P_v$ is sorted by $x$-coordinate (see \cite[Lemma 3]{ref:ChanAl16} for a similar problem). Computing the sorted lists for $P_v$ for all nodes $v\in T$ can be done in $O(n\log n+n\log m)$ time, as described in the proof of Lemma~\ref{lem:60}. As such, constructing $C_v$ takes additional $O(|P_v|)$ time and the total time for constructing $C_v$ of all nodes $v\in T$ is $O(n\log m)$ since $\sum_{v\in T}|P_v|=O(n\log m)$. 

For each disk $s_i\in S$, our task is to determine whether $c_i\in C_v$ for each node $v\in \pi_i$. Note that the boundary of $C_v$ consists of a line segment on $\ell$ bounding $C_v$ from above and an $x$-monotone curve bounding $C_v$ from below. 
The projections of the vertices of $C_v$ onto the $x$-axis $\ell$ partition $\ell$ into a set $\calI_v$ of $O(|P_v|)$ intervals. To determine whether $c_i\in C_v$, it suffices to find the interval of $I_v$ containing $x(c_i)$, the $x$-coordinate of $c_i$, after which whether $c_i\in C_v$ can be decided in $O(1)$ time. Finding the interval of $I_v$ containing $x(c_i)$ can be done in $O(\log n)$ time by binary search. If we do this for all nodes $v\in \pi_i$, the total time to determine whether $s_i$ is prunable would be $O(\log m\log n)$. We can improve the runtime to $O(\log m+\log n)$ using fractional cascading~\cite{ref:ChazelleFr86} in a way similar to the proof of Lemma~\ref{lem:50}. More specifically, we construct a fractional cascading data structure on the intervals of $\calI_v$ of all nodes $v\in T$, which takes $O(n\log m)$ time since the total number of such intervals is linear in $\sum_{v\in T}|P_v|$, which is $O(n\log m)$. With the fractional cascading data structure, we only need to do binary search on the set of the intervals stored at the root of $T$ to find the interval containing $x(c_i)$, which takes $O(\log (n\log m))$ time. After that, following the path $\pi_i$ in a top-down manner, the interval of $\calI_v$ containing $x(c_i)$ for each node $v\in \pi_i$ can be determined in $O(1)$ time~\cite{ref:ChazelleFr86}. As such, whether the disk $s_i$ is prunable can be determined in $O(\log n+\log m)$ time. 

In summary, the total time of the algorithm for finding all prunable disks of $S$ is bounded by $O((n+m)\log (n+m))$. \qed
\end{proof}

Combining Lemmas~\ref{lem:50}, \ref{lem:70}, and \ref{lem:80} leads to the following result for the unit-disk case. 

\begin{theorem}\label{theo:coverageUnit}
Given a set $P$ of $n$ points and a set $S$ of $m$ unit disks in the plane such that the centers of the disks are separated from the points of $P$ by a line, the disk coverage problem for $P$ and $S$ is solvable in $O((n+m)\log(n+m))$ time.
\end{theorem}

\bibliographystyle{plainurl}

\end{document}